\documentclass[reprint,aps,floatfix,pra,superscriptaddress]{revtex4-1}
\usepackage{graphicx}
\usepackage{hyperref}
\usepackage{amsmath}
\usepackage{amssymb}
\usepackage{amsthm}
\usepackage{braket}
\usepackage{dsfont}
\usepackage{color}

\DeclareMathOperator{\Tr}{Tr}
\DeclareMathOperator{\Erf}{Erf}
\DeclareMathOperator{\Id}{\mathds{1}}

\newcommand{\dd}[1]{\textrm d{#1}} 
\newcommand{\norm}[1]{\left\lVert#1\right\rVert}

\newtheorem{theorem}{Theorem}
\newtheorem{proposition}{Proposition}

\theoremstyle{definition}
\newtheorem{exmp}{Example}

\begin{document}
\title{Correlations in local measurements and entanglement in many-body systems}
\author{Chae-Yeun Park}
\affiliation{Asia Pacific Center for Theoretical Physics, Pohang, 37673, Korea}
\author{Jaeyoon Cho}
\affiliation{Asia Pacific Center for Theoretical Physics, Pohang, 37673, Korea}
\affiliation{Department of Physics, POSTECH, Pohang, 37673, Korea}
\date{\today}

\begin{abstract}
While entanglement plays an important role in characterizing quantum many-body systems, it is hardly possible to directly access many-body entanglement in real experiments.
In this paper, we study how bipartite entanglement of many-body states is manifested in the correlation of local measurement outcomes.
In particular, we consider a measure of correlation defined as the statistical distance between the joint probability distribution of local measurement outcomes and the product of its marginal distributions.
Various bounds of this measure are obtained and several examples of many-body states are considered as a testbed for the measure.
We also generalize the framework to the case of imprecise measurement and argue that the considered measure is related to the concept of quantum macroscopicity.
\end{abstract}

\maketitle

\section{Introduction}

Entanglement~\cite{horodecki2009quantum} is a distinctive feature of quantum mechanics, which exposes fundamental differences between quantum and classical physics~\cite{bell1964,clauser1969proposed,brunner2014bell} and can be exploited as a resource for quantum information processing~\cite{NandC}. 
Entanglement is also a useful tool for characterizing quantum states in many-body systems~\cite{amico08,eisert10}.
For example, ground states of gapped Hamiltonians typically follow an area law~\cite{hastings2006,eisert10,brandao13,cho2017simple}, whereas random states follow a volume law of entanglement~\cite{page1993average,hayden2006aspects}.
Amid experimental developments in engineering many-body quantum systems~\cite{bloch2008many,bloch2012quantum,blatt2012quantum,ludlow2015optical}, a great deal of interest has been generated in examining such features of many-body entanglement in real experiments. 
For example, there have been several proposals for measuring R\'enyi $\alpha=2$ entanglement entropies~\cite{ekert2002direct,palmer2005detection,daley2012measuring} and their experimental realizations~\cite{greiner2015,kaufman2016quantum}.
Generally speaking, however, it is very hard to directly measure the entanglement as it is a nonlinear function of the state itself, not an observable.
In order to measure the entanglement, one needs to obtain the density matrix through a quantum state tomography or find the appropriate relations to other measurable quantities, which are nontrivial in many-body systems.

In this paper, we study the many-body entanglement in terms of the correlation in local measurements.
To be specific, we consider a bipartite separation of many-body spin states and positive-operator valued measures (POVMs) acting on each party separately. 
We then investigate the correlation in such local POVM measurements, which is quantified by the statistical distance (total variation distance) between the joint probability distribution of the measurement outcome and the product of its marginal distributions. 
Formally, given a quantum state $\rho_{AB}$ of a composite system $A\otimes  B$ and local POVMs $\{M_i\}$ and $\{N_j\}$ acting on the subsystems $A$ and $B$, respectively, we consider
\begin{align}
	&\Delta_D(\{M_i\},\{N_j\}) \nonumber \\
	&\equiv \frac{1}{2} \sum_{i,j} |\Tr[M_i \otimes N_j (\rho_{AB} - \rho_A \otimes \rho_B)]|,
\end{align}
where $\rho_A=\Tr_B\rho_{AB}$ and $\rho_B=\Tr_A\rho_{AB}$. Letting $P_A(i)=\Tr[(M_i\otimes \Id_B) \rho_{AB}]$, $P_B(j)=\Tr[(\Id_A\otimes N_j)\rho_{AB}]$, and $P_{AB}(i,j)=\Tr[(M_i\otimes N_j)\rho_{AB}]$, this quantity can be written more straightforwardly as
\begin{align}
	\Delta_D(\{M_i\},\{N_j\}) = \frac{1}{2} \sum_{i,j} |P_{AB}(i,j) - P_A(i)P_B(j)|. \label{eq:def2}
\end{align}
For convenience, we will call this quantity a correlation in local measurements (CLM) throughout the paper.

Apparently, for general mixed state $\rho_{AB}$, the CLM does not necessarily capture the entanglement between $A$ and $B$. 
On the other hand, if the state $\rho_{AB}$ is guaranteed to be pure, the CLM should be nonzero for properly chosen POVMs if and only if $\rho_{AB}$ is an entangled state. 
Our aim is to study such relation between the CLM and the entanglement in a {\em quantitative} manner under the condition that $\rho_{AB}$ is a pure many-body spin state. 
Note that by definition, the CLM has a direct relevance to real experimental situations. 
Note also that the CLM is different from conventional correlation functions of two local operators like $\Tr[O_A\otimes O_B(\rho_{AB}-\rho_A\otimes \rho_B)]$ as the CLM is defined by the probability distribution of the measurement outcome, not by the expectation values of general operators. 
There have been earlier works that studied correlation measures involving local measurements~\cite{henderson2001classical,ollivier2001quantum,wu2009correlations,modi2012classical}. 
However, the main focus of them was on investigating quantum correlations that are not captured by local measurements.
Our focus, on the other hand, is on how far one can access the quantum correlation only using local POVM measurements, especially, in many-body systems.

In Sec.~\ref{sec:def}, we investigate the relation between the CLM and other correlation and entanglement measures that have been studied before~\cite{modi2012classical,brodutch12,paula13}.
We then examine, in Sec.~\ref{sec:example}, the CLM for several examples---Haar random states, spin squeezed states, and the ground state of the Heisenberg XXZ model---under the restriction that local measurements are performed in the basis of a collective spin operator.
In Sec.~\ref{sec:eff_imprecision}, we generalize the CLM to the case of imprecise measurement and find its relation to the concept of quantum macroscopicity~\cite{shimizu02,lee2011,frowis2012,park2016}.
We further investigate in Sec.~\ref{sec:bell} how the imprecise measurement affects Bell's inequalities and conclude the paper in Sec.~\ref{sec:conc}.


\section{General properties of the CLM} \label{sec:def}

Before proceeding, it is worthwhile to mention the relation between $\Delta_D$ and another type of correlation measure defined as
\begin{align}
	Cov(A:B) = \max_{M_A,M_B}\frac{|\Tr[M_A\otimes M_B (\rho_{AB} - \rho_A \otimes \rho_B)]|}{\norm{M_A}\norm{M_B}}, \label{eq:cov} 
\end{align}
where the maximization is carried over all operators $M_A$ and $M_B$ acting on subsystems A and B, respectively. Here, $\norm{O}$ is the operator norm of $O$ given by the maximum eigenvalue of $\sqrt{O^\dagger O}$.
The correlation measure $Cov(A:B)$ has been investigated in various contexts~\cite{hastings2006,brandao13,brandao15a,farrelly17,cho2017simple}.
The detailed relation between $Cov(A:B)$ and $\Delta_D$ is not clear. However, when we restrict the maximization in Eq.~\eqref{eq:cov} only to Hermitian operators, it is simple to show that $2\max_{\{M_i\}, \{N_j\}} \Delta_D(\{M_i\}, \{N_j\})$ upper bounds $Cov(A:B)$.

Let us first investigate the relation between $\Delta_D$ and quantum mutual information $I(A:B) = S(\rho_A) + S(\rho_B) - S(\rho_{AB})$, where $S(\rho) = -\Tr[\rho \log \rho]$ is the von Neumann entropy. Throughout the paper, all logarithms will be taken to base $2$.

\begin{proposition} \label{prop:upper_bound_delta}
For a bipartite quantum state $\rho_{AB}$, the following inequality holds for any POVMs:
\begin{align}
	\Delta_D \leq \mathcal{T}(\rho_{AB})\leq \min \Bigl\{ \sqrt{\frac{I(A:B)}{2 \log e}}, \sqrt{1-2^{-I(A:B)}}\Bigr\},
\end{align}
where $\mathcal{T}(\rho_{AB}) = \Tr|\rho_{AB} - \rho_A \otimes \rho_B|/2$ is the total correlation~\cite{groisman05,modi10} measured using the trace distance~\cite{brodutch12,paula13}.
\end{proposition}
\begin{proof}
\begin{align}
	&\Delta_D(\{M_i\}, \{N_j\}) \nonumber \\
	&= \frac{1}{2} \sum_{i,j} |\Tr[M_i \otimes N_j (\rho_{AB}- \rho_A\otimes \rho_B) ]| \nonumber \\
	&\leq \frac{1}{2} \max_{\{K_m\}} \sum_m |\Tr[K_m(\rho_{AB} - \rho_A \otimes \rho_B)]|, \label{eq:p1}
\end{align}
where the maximization is carried over all valid POVMs $\{K_m\}$ for the composite system $A\otimes B$ that satisfy $\sum_m K_m = \Id$ and $K_m \geq 0$ for all $m$. 
The first inequality of the theorem  straightforwardly follows from the fact that the last line in Eq.~\eqref{eq:p1} is nothing but the trace distance $D(\rho_{AB}, \rho_A \otimes \rho_B)$, hence $\mathcal{T}(\rho_{AB})$, where $D(\rho,\sigma)=\Tr|\rho-\sigma|/2$~\cite{NandC}.
The second inequality consists of two parts.
The first part is a well-known Pinsker's inequality, which states $\Tr|\rho_{AB} - \rho_A \otimes \rho_B|/2 \leq \sqrt{I(A:B)/2 \log 2}$~\cite{ohya2004}.
The second part comes from the relations between quantum distances. 
It is known that $D(\rho,\sigma) \leq \sqrt{1-F(\rho,\sigma)^2}$, where $F(\rho,\sigma) = \Tr[\rho^{1/2}\sigma \rho^{1/2}]^{1/2}$ is the fidelity between two quantum states. 
Using the relations between the affinity $A(\rho,\sigma)=\Tr[\rho^{1/2}\sigma^{1/2}]$~\cite{luo04} and other quantities, $A(\rho,\sigma) \leq F(\rho,\sigma)$ and $-\log A(\rho,\sigma) \leq S(\rho || \sigma)/2$~\cite{audenaert12}, the second inequality is obtained. Here, $S(\rho || \sigma) = \Tr[\rho \log \rho - \rho \log \sigma]$ is the relative entropy between $\rho$ and $\sigma$ and $S(\rho_{AB} || \rho_A\otimes\rho_B) = I(A:B)$.
\end{proof}

We note that the Pinsker's inequality is tighter when $I(A:B)$ is smaller, while it is meaningless when $I(A:B) \geq 2 \log e$. 
We also note that there is a previous study~\cite{hall13} that investigated the relation between $\Delta_D$ and $I(A:B)$ for systems of two qubits.

For pure state $\rho_{AB} = \ket{\psi}\bra{\psi}$, $I(A:B) = 2 S(\rho_A)$ is twice the entanglement entropy of $\ket{\psi}$, $S(\rho_A) = -\Tr[\rho_A \log \rho_A]$.
Thus, Proposition~\ref{prop:upper_bound_delta} implies that $\Delta_D$ must be small when the entanglement is small.
Let us further investigate the relation between $\Delta_D$ and the entanglement for $\rho_{AB}$ being pure.
The staring point is a simple proposition.

\begin{proposition} \label{prop:cor_ent}
	A pure quantum state $\ket\psi$ is a separable state of two parties (A and B) $\ket{\psi} = \ket{\phi_A} \otimes \ket{\phi_B}$ if and only if $\Delta_D(\{M_i\},\{N_j\}) = 0$ for any POVMs $\{M_i\}$ and $\{N_j\}$.
\end{proposition}

The question is, what is the lower bound of $\Delta_D(\{M_i\},\{N_j\})$ with an optimal choice of the POVMs when the pure state $\ket\psi$ is entangled? 
The following theorem gives a partial answer.

\begin{theorem}\label{thm:lin_ent_delta}
	For a pure state $\ket{\psi}$, there exist POVMs $\{M_i\},\{N_j\}$ such that $\Delta_D(\{M_i\},\{N_j\}) \geq 1-\mathcal{P}$ where $\mathcal{P} = \Tr[\rho_A^2]$ is the purity of the reduced density matrix.
\end{theorem}
\begin{proof}
	We prove this theorem by explicitly constructing the POVMs.
Suppose that the Schmidt decomposition of $\ket{\psi}$ is given by $\ket{\psi} = \sum_k \sqrt{\lambda_k} \ket{k_A} \ket{k_B}$ with $\sum_k \lambda_k = 1$, where $\lambda_k \geq 0$ are Schmidt coefficients. 
Using the projective measurements in the Schmidt basis $\{M_i = \ket{i_A} \bra{i_A}\}$ and $\{N_j = \ket{j_B} \bra{j_B}\}$, the probability outcomes are given by $P_{AB}(i,j) = \lambda_i \delta_{i,j}$, $P_{A}(i) = \lambda_i$, and $P_B(i) = \lambda_j$.
Here, $\delta_{i,j}$ is the Kronecker delta function.
Then, for these POVMs,
\begin{align}
	\Delta_D &= \frac{1}{2} \sum_{i,j} |P_{AB}(i,j) - P_A(i)P_B(j)| \nonumber \\
	&= \frac{1}{2} \sum_{i,j} |\lambda_i \delta_{i,j} - \lambda_i \lambda_j| \nonumber \\
	&= \frac{1}{2} \Bigl[ \sum_{i} |\lambda_i - \lambda_i^2| + \sum_{i\neq j} \lambda_i\lambda_j \Bigr].
\end{align}
Using $1 = \sum_{i,j} \lambda_i \lambda_j = \sum_{i} \lambda_i^2 + \sum_{i \neq j} \lambda_i \lambda_j$,
we obtain 
\begin{align}
	\Delta_D = 1-\sum_i \lambda_i^2 = 1 - \mathcal{P}.
\end{align}
\end{proof}

From the theorem, $\Delta_D(\{M_i\},\{N_j\}) = 0$ for all POVMs implies $\mathcal{P}=1$, which means $\rho_A$ is pure and hence $\rho_{AB}$ is separable. 
Note that the lower bound $1-\mathcal{P}$ is the linear entropy, which has been widely investigated in quantum information theory. 
The linear entropy is a nice indication of entanglement for pure states, although it is not an entanglement monotone in general.

\section{CLM for collective spin measurements} \label{sec:example}

So far, our discussion was general; we did not consider any specific form of POVMs or a system. 
In this section, we consider several examples of many-body spin systems to investigate the properties of the CLM.
To be specific, we consider systems of $N$ $s=1/2$ spins with its subsystems $A$ and $B$ each containing $N/2$ spins.
As a natural choice, we consider the case wherein each party performs a collective spin measurement.
For subsystems $A$, the spins are measured in the basis of $S_A(\hat{\alpha})=\hat{\alpha}\cdot \pmb{S}_A$, where $\hat{\alpha}$ is a unit vector and $\pmb{S}_A = \sum_{i \in A} \pmb{\sigma}^{(i)}/2$ is the collective spin operator. 
Here, $\pmb{\sigma}^{(i)} = \{\sigma^{(i)}_x,\sigma^{(i)}_y,\sigma^{(i)}_z\}$ is the vector of Pauli spin operators for the $i$-th spin.
We can obtain the POVM for $S_A(\hat{\alpha})$ from the decomposition
\begin{align}
	S_A(\hat{\alpha}) = \sum_{i=-N/4}^{N/4} i \sum_{\mu_i } \ket{i, \mu_i} \bra{i, \mu_i}
\end{align}
where $i \in [-N/4,N/4]$ are possible measurement outcomes and $\mu_i$ is the index for the degenerate subspace corresponding to the outcome $i$.
Then the POVM can be written as $M_i(\hat{\alpha}) = \sum_{\mu_i} \ket{i,\mu_i}\bra{i,\mu_i}$.
Likewise, we also define $S_B(\hat{\beta}) = \hat{\beta} \cdot \pmb{S}_B$ and the corresponding POVM $\{N_j(\hat{\beta})\}$ such that $S_B(\hat{\beta}) = \sum_{j = -N/4}^{N/4} j N_j(\hat{\beta})$ for subsystem B. 
To simplify the notation, the shorthand expression $\Delta_D(\hat{\alpha},\hat{\beta})$ will be used throughout this section to designate $\Delta_D(\{M_i (\hat{\alpha})\}, \{N_j (\hat{\beta})\})$ unless it confuses.

Before proceeding, let us first consider simple heuristic examples. 

\begin{exmp}\label{exmp:exmp1}
Let $\ket{\psi_0} = (\ket{\downarrow}^{\otimes N} + \ket{\uparrow}^{\otimes N})/\sqrt{2}$ and $\ket{\psi_1} = (\ket{\uparrow}^{\otimes N-1}\ket{\downarrow} + \ket{\downarrow}\ket{\uparrow}^{\otimes N-1})/\sqrt{2}$. Then $\Delta_D(\hat{z},\hat{z})=1/2$ for both the states. The possible outcome pairs $(i,j)$ from the measurements are $\{(N/4,N/4),(-N/4,-N/4)\}$ and $\{(N/4, N/4-1),(N/4-1,N/4)\}$, respectively. For the same states, correlation function $\braket{S_A(\hat{z})\otimes S_B(\hat{z})}-\braket{S_A(\hat{z})}\braket{S_B(\hat{z})}$ yields $N^2/16$ and $-1/4$, respectively, which largely differ. This example illustrates a stark difference between the CLM and the correlation function.
\end{exmp}

\begin{exmp}
Let us consider $\ket{\psi_2} = C \sum_{P} (P \ket{\downarrow}^{\otimes N/4}\ket{\uparrow}^{\otimes N/4})(P \ket{\downarrow}^{\otimes N/4}\ket{\uparrow}^{\otimes N/4})$, where the summation is over all possible permutations $P$. The normalization constant $C$ is given by $C = {{N/2}\choose{N/4}}^{-1/2}$. As the whole component states live in the subspace of $S_A(\hat{z})=S_B(\hat{z})=0$, we can see that $P_{AB}(i,j) = \delta_{i,0}\delta_{j,0}$ and $P_A(i) = \delta_{i,0}$, $P_B(j) = \delta_{j,0}$. Therefore, $\Delta_D(\hat{z},\hat{z}) = 0$. On the other hand, when we compute the entanglement entropy, we get $S=\log {{N/2} \choose {N/4}}$. Using Stirling's formula, this can be approximated as $S \approx N/2 \ln 2 + \mathcal{O}(\log N)$ for $N \gg 1$, which indicates that the entanglement is extensive. 
This result illustrates that $\Delta_D$ using collective spin measurements cannot capture entanglement of some states.
\end{exmp}

\subsection{Random States}

In this subsection, we investigate the behavior of the CLM optimized over all directions, i.e., $\max_{\hat{\alpha},\hat{\beta}}\Delta_D(\hat{\alpha}, \hat{\beta})$, for Haar random states. 
For this, recall Levy's lemma which implies that the values of a Lipschitz continuous function $f$ are all concentrated to its mean value $\braket{f}$. 
Formally it is written as follows.
\begin{theorem}[Levy's lemma; see Ref.~\cite{ledoux2005concentration}]
Let $f:\mathbb{S}^k \rightarrow \mathbb{R}$ be a function with Lipshitz constant $\eta$ and $\phi \in \mathbb{S}^k$ be a point chosen uniformly at random. Then, 
\begin{align}
	{\rm Pr}\bigl[|f(\phi) - \braket{f}\bigr| > \epsilon] \leq 2 \exp (-2C(k+1)\epsilon^2/\eta^2)
\end{align}
for a constant $C > 0$ that may be chosen as $C=(18 \pi^3)^{-1}$.
\label{thm2}
\end{theorem}

\begin{figure}[t]
	\centering
	\resizebox{0.48\textwidth}{!}{
		\includegraphics{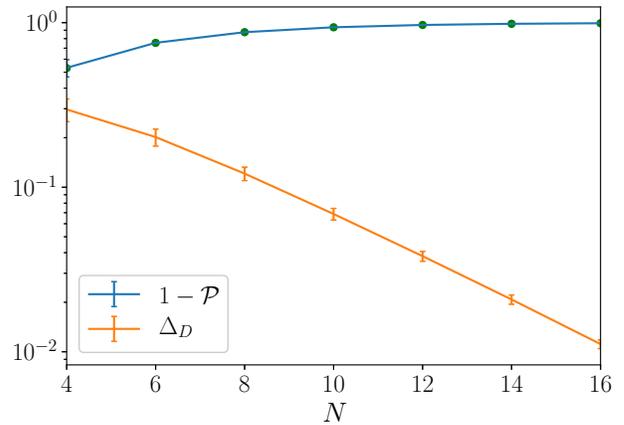}
	}
	\caption{\label{fig:RandomStates}
	Optimal CLM for collective spin measurements and linear entropy of a subsystem, obtained for Haar random states. For each $N$, $10^3$ random states were taken and the results were averaged. The green dots represent the analytic values of the average linear entropy.
	}
\end{figure}

We now prove the Lipschitz continuity of the optimized CLM.

\begin{theorem}
	$\max_{\hat{\alpha},\hat{\beta}}\Delta_D(\hat{\alpha}, \hat{\beta})$ is a Lipschitz continuous function of $\ket{\psi}$ with the Lipschitz constant $\eta \leq 12$.
\label{thm3}
\end{theorem}

\begin{proof}
	Let $\ket{\psi}$ and $\ket{\psi'}$ be two different pure states. Then the difference of $\Delta_D$ is given by
\begin{align*}
	&\Bigl|\frac{1}{2} \sum_{i\in I,j \in J} |\Tr[M_i \otimes N_j (\rho_{AB}-\rho_A \otimes \rho_B)]|\\
	&-\frac{1}{2} \sum_{i \in I',j \in J'} |\Tr[M_i' \otimes N_j' (\rho_{AB}'-\rho_A' \otimes \rho_B')]|\Bigr|,
\end{align*}
where $\rho_{AB} = \ket{\psi}\bra{\psi}$ and $\rho_{AB}' = \ket{\psi'}\bra{\psi'}$.
$M_i$ and $N_j$ are the POVMs that maximize $\Delta_D$ for $\ket{\psi}$, and likewise for $M_i'$ and $N_j'$.
Using the triangular inequality, the above expression is bounded by
\begin{align*}
	&\leq \frac{1}{2} \sum_{i\in I \sqcup I' ,j \in J \sqcup J'} \Bigl|\Tr[(M_i \otimes N_j) \rho_{AB}-(M_i'\otimes N_j')\rho_{AB}']\\
	&\quad \quad -\Tr[(M_i' \otimes N_j') (\rho_A'\otimes \rho_B')-(M_i \otimes N_j)(\rho_A \otimes \rho_B))]\Bigr|.
\end{align*}
Let us define $K_i$ such that $K_i = M_i/2$ for $i \in I$ and $K_i = M_i'/2$ for $i \in I'$. Then $\{K_i\}$ is a valid POVM defined for $i \in I \sqcup I'$. We also define $L_j$ which are $N_j/2$ for $j \in J$ and $N_j'/2$ for $j \in J'$.
Using these POVMs, we obtain
\begin{align*}
	&=2 \sum_{i\in I \sqcup I' ,j \in J \sqcup J'} \Bigl|\Tr[(K_i \otimes L_j) (\rho_{AB}-\rho_{AB}')]\\
	&\quad \quad -\Tr[(K_i \otimes L_j) (\rho_A'\otimes \rho_B' - \rho_A \otimes \rho_B)]\Bigr| \\
	&\leq 2 \bigl[\Tr|\rho_{AB} - \rho_{AB}'| + \Tr|\rho_{A}\otimes\rho_{B}-\rho_A'\otimes \rho_B'|\bigr],
\end{align*}
where we have again used $\Tr|\rho-\sigma| = \max_{\{K_m\}}|K_m(\rho-\sigma)|$ to obtain the last inequality. Moreover, $\Tr|\rho_A\otimes\rho_B - \rho_A'\otimes \rho_B'| \leq \Tr|\rho_A\otimes(\rho_B - \rho_B')| + \Tr|(\rho_A-\rho_A')\otimes \rho_B'|  \leq 2 \Tr |\rho_{AB} - \rho_{AB}'|$. To sum up,
\begin{align*}
	& \leq 6 \Tr|\rho_{AB}-\rho_{AB}'| = 6 \Tr|\ket{\psi}\bra{\psi} - \ket{\psi'}\bra{\psi'}|\\
	&= 12 \sqrt{1 - |\braket{\psi|\psi'}|^2} \leq 12\norm{\ket{\psi} - \ket{\psi'}}_2.
\end{align*}
Therefore, the Lipschitz constant $\eta \leq 12$ is obtained.
\end{proof}

The above two theorems imply that as $N\rightarrow\infty$, the optimal CLMs for Haar random states converges to a certain value with a vanishing variance. 
We numerically generated Haar random states and obtained the optimal vectors $\{\hat\alpha,\hat\beta\}$ maximizing the CLM for {\em each} given state.
The result, averaged over $10^3$ random states, is plotted in Fig.~\ref{fig:RandomStates} along with the linear entropy of a subsystem.
It show that while the linear entropy of a subsystem increases with $N$ and coincides with the analytic result $1-\braket{\mathcal{P}}=1-2^{N/2+1}/(2^N +1)$, the optimal CLM decreases exponentially with $N$.
The collective spin measurement is thus inappropriate to capture the entanglement of random states~\cite{page1993average,hayden2006aspects}. 
This is the case even if we consider more general collective spin bases $S_A(\{\hat\alpha_i\})=\sum_{i\in A}\hat\alpha_i\cdot \pmb{\sigma}^{(i)}/2$ and $S_B(\{\hat\beta_i\})=\sum_{i\in B}\hat\beta_i\cdot \pmb{\sigma}^{(i)}/2$ and optimize the CLM over all unit vectors $\{\hat\alpha_i,\hat\beta_i\}$. 
This can be understood as follows. 
For given random state $\ket\psi$ and the corresponding optimal measurement bases $\{\ket{\alpha} = \ket{i, \mu_i}\}$ and $\{\ket{\beta} = \ket{j, \mu_j}\}$ for subsystems $A$ and $B$, respectively, with $-N/4\le i,j\le N/4$, one can write the state as $\ket{\psi} = \sum_{\alpha,\beta} A_{\alpha,\beta} \ket{\alpha}\ket{\beta}$.
It is known that as $N\rightarrow\infty$, $|A_{\alpha,\beta}|^2$ should approach $1/2^N$ with a vanishing fluctuation. 
In such a limit, $P_{AB}(i,j) = {N/2 \choose i+N/4}{N/2 \choose j+N/4}/2^N$ and $P_{A}(i) = P_{B}(i) = {N/2 \choose i+N/4}/2^{N/2}$, leading to $P_{AB}(i,j) = P_A(i)P_B(j)$ and hence vanishing $\Delta_D$.

\begin{figure}[t]
	\centering
	\resizebox{0.45\textwidth}{!}{
		\includegraphics{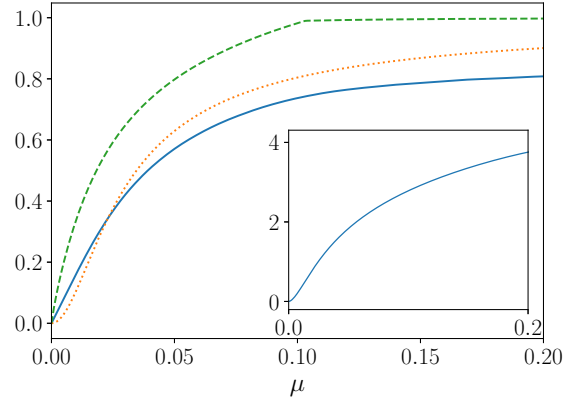}
	}
	\caption{\label{fig:Squeezing}
	$\Delta_D(\hat{z},\hat{z})$ (blue curve), linear entropy of a subsystem (red dotted curve), and the upper bound from Proposition~\ref{prop:upper_bound_delta} (green dashed curve), obtained for spin squeezed states $V_\mu \ket{+}^{\otimes N}$ as a function of the squeezing strength $\mu$. 
	The system size is $N=200$.
	The inset shows the entanglement entropy for comparison.
	}
\end{figure}

\subsection{Spin Squeezed States}

In this subsection, we consider one-axis twisted states that are generated by applying a squeezing operator
\begin{align}
	V_{\mu} = e^{-i\nu S_x} e^{-i \mu S_z^2 /2}
\end{align}
to the spin coherent state in $x$-direction $\ket{+}^{\otimes N}$, where $\ket{+} = (\ket{\uparrow} + \ket{\downarrow})\sqrt{2}$~\cite{kitagawa1993squeezed} (for a review, see Ref.~\cite{ma2011quantum}).
Here, $S_x = S_A(\hat{x})+S_B(\hat{x})$ and $S_z = S_A(\hat{z})+S_B(\hat{z})$. 
This kind of squeezed states have been experimentally generated in many different set-ups~\cite{meyer2001experimental,gross2010nonlinear,riedel2010atom,bohnet2016quantum}.

As spin coherent states and squeezing operators are symmetric under any permutations between spins, the resulting squeezed states also live in a permutation symmetric subspace of the total Hilbert space. 
One may use a vector space spanned by Dicke states to efficiently represent this state.
Dicke states are given by
\begin{align}
	\ket{N, k}	= {N \choose k}^{-1/2} \sum_P P(\ket{\uparrow}^{\otimes k}\ket{\downarrow}^{\otimes N-k})
\end{align}
for $0 \leq k \leq N$, where the summation runs over all possible permutations. 
It is easy to show that when we divide a subspace generated by Dicke states into two subsystems of $N/2$ spins, Dicke states in each subsystem ($\ket{N/2, k}$) also become a basis set, i.e. $\ket{N,k} = \sum_{r=0}^{k} C_r \ket{N/2,r}_A \ket{N/2,k-r}_B$.
Consequently, the entanglement entropy of any permutation symmetric state is upper bounded by $\log(N/2+1)$.

Expectation values and the variances of spin operators for the spin squeezed state $V_\mu \ket{+}^{\otimes N}$ are calculated in Ref.~\cite{kitagawa1993squeezed}. 
It shows
\begin{align*}
	&\braket{S_x} = \frac{N}{2} \cos^{N-1} \frac{\mu}{2}, \quad \braket{S_y} = \braket{S_z} = 0, \\
	&\braket{\Delta S_x^2} = \frac{N}{4}\bigl[N(1-\cos^{2(N-1)}\frac{\mu}{2}-\frac{N-1}{2}A\bigr],\\
	&\braket{\Delta S_{y,z}^2} = \frac{N}{4} \bigl\{ 1+ \frac{N-1}{4} [A \pm \sqrt{A^2 + B^2} \cos(2\nu + 2 \delta)]\bigr\},
\end{align*}
where $A = 1-\cos^{N-2}\mu$, $B = 4 \sin \frac{\mu}{2} \cos^{N-2} \frac{\mu}{2}$, and $\delta = \frac{1}{2}\arctan \frac{B}{A}$. 

For the system size $N=200$, we performed numerical calculations for $\nu = \frac{\pi}{2} - \delta$ that maximizes $\braket{\Delta S_z^2}$ and minimizes $\braket{\Delta S_y^2}$.
In Fig.~\ref{fig:Squeezing}, $\Delta_D(\hat{z},\hat{z})$ and the linear entropy of a subsystem are plotted with respect to the squeezing strength $\mu$. 
For comparison, the upper bound of the CLM from Proposition~\ref{prop:upper_bound_delta} and the entanglement entropy are also plotted.
All those results show similar functional behaviors, suggesting that the CLM is appropriate to capture the entanglement in this case.
One may compare $\Delta_D(\hat{z},\hat{z})$ with the value for the GHZ state ($\ket{\psi_0}$ in Example~\ref{exmp:exmp1}), for which $\Delta_D=0.5$.
We find that $\Delta_D(\hat{z},\hat{z}) \geq 0.5$ for $\mu \gtrsim 0.04$.

\subsection{Ground States of the Heisenberg XXZ Model}

As a final example, we consider the ground state of the one-dimensional Heisenberg XXZ model.
The Hamiltonian of the model is given by
\begin{align*}
	H = \sum_{i=1}^N \bigl[ J(\sigma^{(i)}_x\sigma^{(i+1)}_x + \sigma^{(i)}_y \sigma^{(i+1)}_y) + J_z \sigma^{(i)}_z \sigma^{(i+1)}_z \bigr],
\end{align*}
where $J > 0$ is the interaction strength and $J_z/J$ determines the strength of anisotropy. 
It is well known that this model is solvable using the Bethe ansatz. 
For $J>0$, the model is gapless in thermodynamic limit ($N \rightarrow \infty$) for $-1 < J_z/J \leq 1$. 
When $J_z/J < -1$, two degenerate ground states are $\ket{\uparrow}^{\otimes N}$ and $\ket{\downarrow}^{\otimes N}$.
As there is no spontaneous symmetry breaking for finite $N$, we take $(\ket{\uparrow}^{\otimes N}+\ket{\downarrow}^{\otimes N})/\sqrt{2}$, which is the GHZ state we have studied in Example~\ref{exmp:exmp1}, as the ground state for $J_z/J < -1$.
For $J_z/J > 1$, the model shows the gapped anti-ferromagnetic phase~\cite{giamarchi2004}. 
The quantum phase transition at $J_z/J=-1$ is the first order and the infinite order Kosterlitz-Thouless transition occurs at $J_z/J = 1$. We note that this Hamiltonian models some real materials~\cite{mikeska2004one} and is implementable using engineered systems such as optical lattices~\cite{duan2003controlling} and trapped ions~\cite{hauke2010complete,bermudez2017long} (see also Ref.~\cite{hazzard2014quantum} which provides the summary of theoretical proposals and experiments of this model).

\begin{figure}[t]
	\centering
	\resizebox{0.48\textwidth}{!}{
		\includegraphics{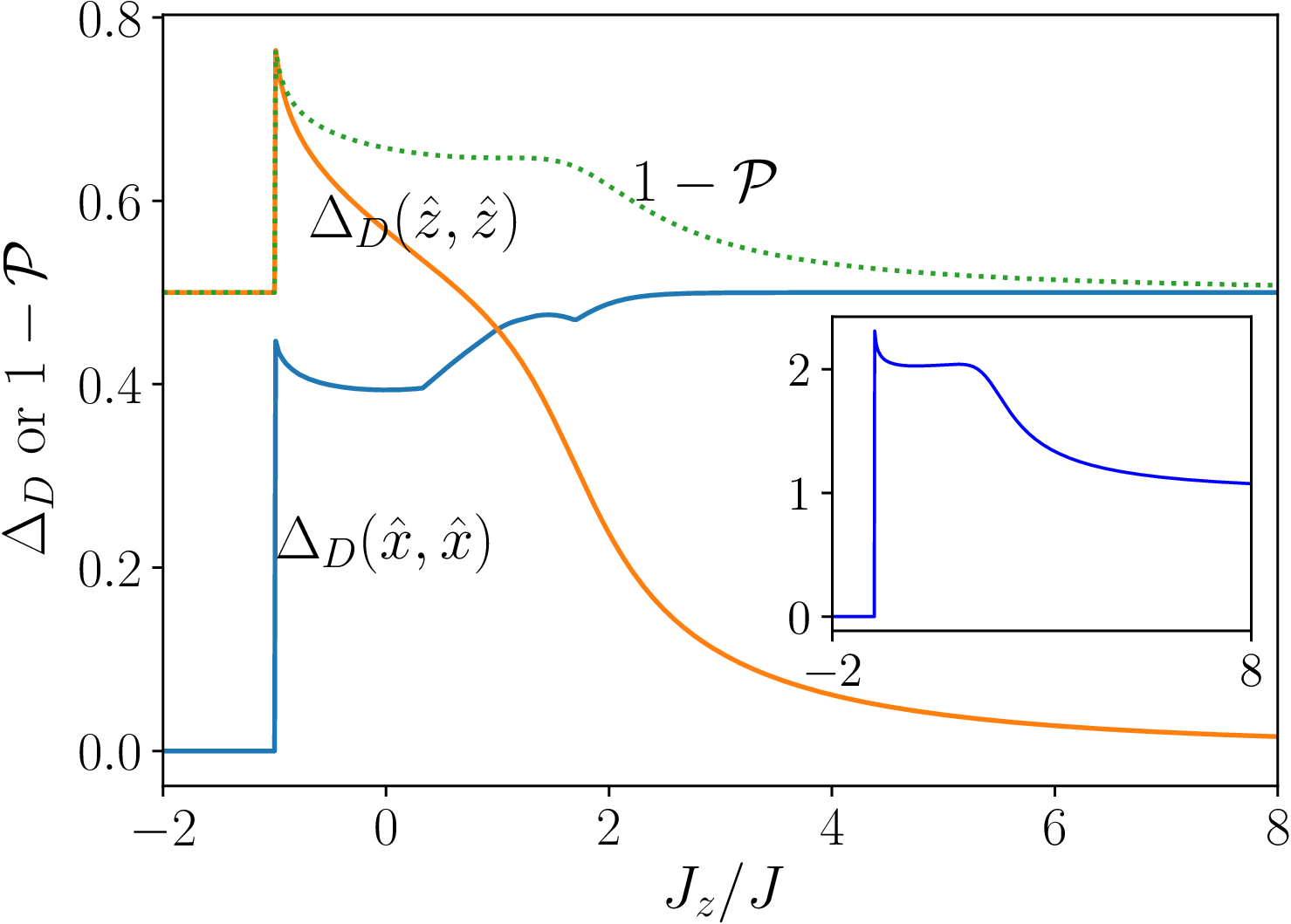}
	}
	\caption{\label{fig:SpinXXZ}
		CLMs in $x$ and $z$ directions and linear entropy $1-\mathcal{P}$ of a subsystem for the ground state of the Heisenberg XXZ model.
		The inset shows the entanglement entropy for comparison.
	}
\end{figure}

\begin{figure}[t]
	\centering
	\resizebox{0.48\textwidth}{!}{
		\includegraphics{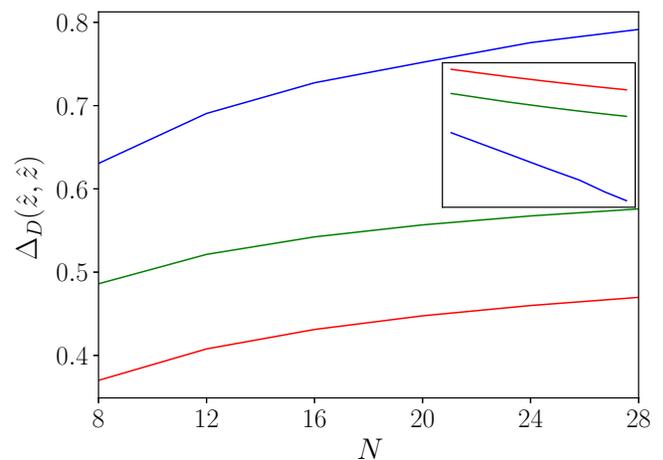}
	}
	\caption{\label{fig:SpinXXZ_nvsd}
		CLM in $z$ direction for the ground state of the Heisenberg XXZ model for $J_z/J = -1^{-}$, $0$, and $1$ (from top to bottom) as a function of $N$. Inset: log-log plot of $N$ versus $1-\Delta_D$, which suggests $\Delta_D \approx 1-c N^{-\alpha}$ scaling.
	}
\end{figure}

For the system size $N=24$, we obtained the ground state using the Lanczos method. 
In Fig.~\ref{fig:SpinXXZ}, $\Delta_D$ in $x$ and $z$ directions and the linear entropy $1-\mathcal{P}$ are plotted for $-2 \leq J_z/J \leq 8$.
We have obtained $\Delta_D(\hat{z},\hat{z}) \geq 0.5$ ($\Delta_D$ for the GHZ state) for $-1.0 < J_z/J \lesssim 0.66$.
The first order phase transition at $J_z/J = -1$ is directly seen from the sudden changes of $\Delta_D$ and $1-\mathcal{P}$.
There is a crossing of $\Delta_D$s in $x$ and $z$ directions at $J_z/J = 1$ as the system has a full $SU(2)$ symmetry at that point. 
Some singular points in $\Delta_D(\hat{x},\hat{x})$ that are nothing to do with a quantum phase transition appear near $J_z/J \approx 0.3$ and $\approx 1.7$.

When $J_z/J \gg 1$, the ground state is the superposition of two N\'{e}el ordered states $\ket{\uparrow\downarrow \cdots} + \ket{\downarrow\uparrow \cdots}$. 
The joint probability distribution of the measurement in $z$ direction is given by $P_{AB}(i,j) = \delta_{i,0}\delta_{j,0}$. 
In this case, $\Delta_D(\hat{z},\hat{z}) = 0$ is obtained and this is consistent with the result in Fig.~\ref{fig:SpinXXZ}. 
By rotating the state, we can also obtain the probability distribution for the measurement in $x$ direction. 
A simple calculation yields $P_{AB}(i,j) = {N/2 \choose i+N/4} {N/2 \choose j+N/4}/2^{N-1}$ when $i+j+N/2$ is even and $P_{AB}(i,j)=0$ otherwise. 
Using this, $\Delta_D(\hat{x},\hat{x}) = 1/2$ is obtained, which also agrees with our numerical result.

We also numerically obtained $\Delta_D (\hat{z},\hat{z})$ at $J_z/J = -1^{+}$, $0$, and $1$ for the system sizes $N$ that are multiples of $4$, which are plotted in Fig.~\ref{fig:SpinXXZ_nvsd}.
These values of $N$ are used as the ground states are translation invariant, i.e., $T \ket{\rm GS}_N = \ket{\rm GS}_N$ (for even $N$ that is not a multiple of $4$, $T \ket{\rm GS}_N = -\ket{\rm GS}_N$).
The result shows that $\Delta_D (\hat{z},\hat{z})$ is increasing with $N$. 
This indicates that a relatively large value of CLM can be obtained for any system size. 
We also find that this increasing behavior follows a power law that is typical for critical systems.

\section{Effects of measurement imprecisions} \label{sec:eff_imprecision}

In practice, any measurement in experiments is imperfect to some degree.
Then, the measurement outcomes are not perfectly discriminated and the CLM $\Delta_D(\{M_i(\hat{\alpha})\},\{N_j(\hat{\beta})\})$ is thus poorly defined.
This motivates us to consider the cases wherein the collective spin measurement of $S_A(\hat{\alpha})$ and $S_B(\hat{\beta})$ has a finite resolution.
For subsystem $A$, the Kraus operators for this type of measurement can be written as~\cite{poulin2005}
\begin{align}
	E^{\sigma}(\hat{\alpha};x) = \sum_{i=-N/4}^{N/4} \sqrt{p^\sigma(x,i)} E_i(\hat{\alpha}),
\end{align}
where $\{E_i(\hat{\alpha})\}$ are the Kraus operators for $M_i(\hat{\alpha})$, given by $M_i(\hat{\alpha})=E_i(\hat{\alpha})^\dagger E_i(\hat{\alpha})$. 
In our case, $E_i(\hat{\alpha})=M_i(\hat{\alpha})$ as $M_i(\hat{\alpha})$ is a projection operator.
Here, $p^\sigma(x,i)$ is a smoothing function, which is a probability distribution function of  continuous variable $x$, i.e., $\int_{x \in D_A} \dd{x} \, p^\sigma(x,i) = 1$. 
The probability $p^\sigma(x,i) \dd{x}$ means the probability to obtain measurement outcomes in $[x, x+dx]$ when the state is actually $i$. 
Here, $\sigma$ is a parameter which determines the resolution of the measurement. 
The Gaussian (normal) distribution $p^\sigma(x,i) = e^{-(x-i)^2/2\sigma^2}/\sqrt{2 \pi \sigma^2}$ with $x\in\mathbb{R}$ is widely used.
Using the Kraus operators, the POVM of continuous outcomes is defined as $M^\sigma (\hat{\alpha};x) = E^{\sigma}(\hat{\alpha};x)^\dagger E^{\sigma}(\hat{\alpha};x)$.
For subsystem $B$, we similarly define the Kraus operator $F^\sigma(\hat{\beta}, y) = \sum_{j=-N/4}^{N/4} \sqrt{p^\sigma (x,j)} F_j(\hat{\beta})$ and the corresponding POVM $N^\sigma (\hat{\beta};y)=F^\sigma(\hat{\beta}, y)^\dagger F^\sigma(\hat{\beta}, y)$. 
This kind of measurement is also called a coarse-grained measurement~\cite{kofler2007}.

The CLM $\Delta_D$ we have used above is defined for measurements with discrete outcomes. 
We define a continuous version of the CLM as
\begin{align}
	&\Delta_C (\{M(x)\}, \{N(y)\}) \nonumber\\
	&\quad \quad = \frac{1}{2} \int_{x \in D_A}\int_{y \in D_B} \dd{x} \dd{y} |P_{AB}(x,y) - P_A(x) P_B(y)|, \label{eq:delta_C}
\end{align}
where $D_A, D_B \subset \mathbb{R}$ are the domains of the possible measurement outcomes for subsystems A and B, respectively. 
Here, the probability distribution functions are given by $P_{AB}(x,y) = \Tr[ M(x) \otimes N(y) \rho_{AB}]$, $P_A(x) = \Tr[M(x) \rho_A]$, and $P_B(y) = \Tr[ N(y) \rho_B]$.
We note that the properties of $\Delta_D$ derived in Sec.~\ref{sec:def} remain valid for $\Delta_C$ as a POVM with continuous outcomes can be reduced to that with discrete outcomes as far as the system is finite dimensional~\cite{chiribella07}.

\begin{exmp}
Let us recall $\ket{\psi_0}$ and $\ket{\psi_1}$ from Example~\ref{exmp:exmp1}. 
A simple calculation yields $\Delta_C(\{M^\sigma(\hat{z};x)\},\{N^\sigma(\hat{z};y)\}) = \Erf(N/(4\sqrt{2}\sigma))^2/2$ for $\ket{\psi_0}$ and $\Erf(1/(2\sqrt{2}\sigma))^2/2$ for $\ket{\psi_1}$ when we use the Gaussian smoothing function. 
Here, $\Erf(x) = \int_{-x}^x e^{-t^2}dt /\sqrt{\pi}$ is the error function. 
Therefore, for large $N \gg 1$, the correlation of $\ket{\psi_0}$ is detectable even with imprecise measurement but that of $\ket{\psi_1}$ is not. 
For instance, when $N=20$ and $\sigma = 2.0$, $\Delta_C \approx 0.488$ for $\ket{\psi_0}$, but $\Delta_C \approx 0.019$ for $\ket{\psi_1}$. 
We also note that when $\sigma\rightarrow 0^+$, $\Delta_C\rightarrow 0.5$ for both states, recovering $\Delta_D$ in Example~\ref{exmp:exmp1}.
\end{exmp}

Our main point of this section is that the CLM with coarse-grained measurements is related to the concept of quantum macroscopicity. 
The following two theorems make the relation more explicit.

\begin{theorem}[Correlation-disturbance] \label{thm:correl-disturb}
\begin{align}
\Delta_C(\{M^\sigma(x)\},\{N^\sigma(y)\}) \leq 1 - \mathcal{F}(\ket{\psi}, \rho_{AB}')^2,
\end{align}
where $\mathcal{F}(\ket{\psi}, \rho) = \braket{\psi|\rho|\psi}^{1/2}$ is the fidelity between a pure state $\ket{\psi}$ and a mixed state $\rho$. Here, $\rho_{AB}'$ is the post-measurement state given by
\begin{align*}
&\rho_{AB}' =\\
& \int_{D_X}dx \int_{D_Y} dy \, E^\sigma (x)\otimes F^\sigma(y)\ket{\psi}\bra{\psi}E^\sigma (x)\otimes F^\sigma(y).
\end{align*}
\end{theorem}
\begin{proof}
Using $|f(x) - g(x)| = \max[f(x), g(x)] - \min[f(x), g(x)]$ and $f(x) + g(x) = \max[f(x), g(x)] + \min[f(x), g(x)]$, we obtain $|P_{AB}(x,y) - P_A(x)P_B(y)| = P_{AB}(x,y) + P_A(x)P_B(y) - 2\min\{P_{AB}(x,y), P_A(x)P_B(y)\}$.
Integrating both sides, we obtain
\begin{align*}
&\Delta_C(\{M^\sigma(x)\},\{N^\sigma(y)\}) \\
&\leq 1 - \int_{D_X} \dd{x} \int_{D_Y} \dd{y}\,\min[P_{AB}(x,y),P_A(x)P_B(y)].
\end{align*}
The theorem follows from
\begin{align}
	P_{AB}(x,y) &= \braket{\psi | [E^{\sigma}(x) \otimes F^{\sigma} (y)] ^2|\psi} \nonumber \\
			   &\geq |\braket{\psi | E^{\sigma}(x) \otimes F^{\sigma} (y) |\psi}|^2 \label{eq:low_ab}
\end{align}
and
\begin{align}
	P_{A}(x)P_{B}(y) &= \braket{\psi | [E^{\sigma}(x) \otimes \Id] ^2|\psi}\braket{\psi | [\Id\otimes F^{\sigma}(y)]^2 |\psi} \nonumber \\
			   &\geq |\braket{\psi | E^{\sigma}(x) \otimes F^{\sigma} (y) |\psi}|^2. \label{eq:low_a_b}
\end{align}
We have used $\braket{\psi|A^2|\psi} \geq \braket{\psi|A|\psi}^2$ for Hermitian $A$ to obtain the inequality in Eq.~\eqref{eq:low_ab} and the Cauchy-Schwartz inequality $\braket{f|f}\braket{g|g} \geq |\braket{f|g}|^2$ with $\ket{f} = [E^\sigma(x)\otimes \Id]\ket{\psi}$ and $\ket{g} = [\Id\otimes F^\sigma(y)]\ket{\psi}$ for the inequality in Eq.~\eqref{eq:low_a_b}.
\end{proof}

\begin{theorem}\label{thm:fidelity-variance}
	For the Gaussian smoothing $p^\sigma(x,i) = e^{-(x-i)^2/2\sigma^2}/\sqrt{2 \pi \sigma^2}$, 
\begin{align*}
	& \mathcal{F}(\ket{\psi}, \rho_{AB}')^2 \\
	& \geq \exp \left(-\frac{{\cal V}_{\ket{\psi}}(S_A(\hat{\alpha})\otimes \Id)+{\cal V}_{\ket{\psi}}(\Id \otimes S_B(\hat{\beta}))}{4\sigma^2}\right),
\end{align*}
where ${\cal V}_{\ket{\psi}}(A) = \braket{\psi|A^2|\psi}-\braket{\psi|A|\psi}^2$ is the variance of operator $A$ for quantum state $\ket{\psi}$.
\end{theorem}

The proof of the theorem can be found in the Appendix. 
The steps for the proof are basically the same as those of Theorem~2 in Ref.~\cite{kwon2017}.
We note that ${\cal V}_{\ket{\psi}}(S_A(\hat{\alpha})\otimes \Id)+{\cal V}_{\ket{\psi}} (\Id \otimes S_B(\hat{\beta}))$ in the theorem has an obvious relation to the measure of quantum macroscopicity defined as 
\begin{align}
	\mathcal{M}(\ket{\psi}) = \max_{A \in S}{\cal V}_{\ket{\psi}}(A),
\end{align}
where $S$ is the set of collective observables given by
\begin{align*}
	S = \Bigl\{ \sum_i \hat{\alpha}^{(i)} \cdot \pmb \sigma^{(i)} : |\hat{\alpha}^{(i)}| = 1 \text{ for all }i \in 1 , \cdots, N \Bigr\}.
\end{align*}
The first definition of this measure appeared in Ref.~\cite{shimizu02} and the measure has been developed in various contexts~\cite{frowis2012,park2016} (see also Ref.~\cite{frowis2017} for a recent review).
As ${\cal V}_{\ket{\psi}}(S_A \otimes \Id) + {\cal V}_{\ket{\psi}}(\Id \otimes S_B) \leq \max\{ {\cal V}_{\ket{\psi}}(S_A\otimes \Id + \Id \otimes S_B), {\cal V}_{\ket \psi}(S_A \otimes \Id - \Id \otimes S_B) \}$ and $2(S_A \otimes \Id \pm \Id \otimes S_B) \in S$, it is evident that ${\cal V}_{\ket \psi}(S_A \otimes \Id) + {\cal V}_{\ket \psi}(\Id \otimes S_B) \leq \mathcal{M}(\ket{\psi})/4$.
Using this result, we can rewrite Theorem~\ref{thm:correl-disturb} as
\begin{align}
\Delta_C(\{M^\sigma(\hat{\alpha};x)\},\{N^\sigma(\hat{\beta};y)\}) \leq 1 - \exp\left(-\frac{\mathcal{M}(\ket{\psi})}{16 \sigma^2}\right). \label{eq:macro_bound}
\end{align}

Previous studies of quantum macroscopicity in many-body spin systems have shown that a class of quantum states of $N$ spins $\ket{\psi_N}$ can be regarded as a macroscopic superposition if $\mathcal{M}(\ket{\psi}_N) = \mathcal{O}(N^2)$, whereas it cannot be if $\mathcal{M}(\ket{\psi}_N) = \mathcal{O}(N)$~\cite{shimizu02,frowis2012,park2016}.
For example, a product state is not a macroscopic superposition as it gives $\mathcal{M}(\ket{\phi_1 \phi_2 \cdots \phi_N}) = N$.
More recent studies have shown that Haar random states~\cite{tichy2016,oszmaniec2016} and asymptotic states in non-integrable systems that thermalize also show an $\mathcal{M} = \mathcal{O}(N)$ behavior~\cite{park2016disappearance}. 
Our result thus implies that the correlations of those latter states with $\mathcal{M} = \mathcal{O}(N)$ cannot be detected if $\sigma \gg \sqrt{N}$.
In some literatures~\cite{kofler2008,sekatski2014,barnea2017}, a course-grained measurement with  $\sigma \gg \sqrt{N}$ is considered as a classical measurement in the sense that the measurement hardly disturb the state for large but finite $N$~\cite{poulin2005,kofler2007,frowis2016}.
Following this line of arguments, our results suggest that the correlation of pure entangled states $\ket{\psi}$ cannot be captured with classical measurements if $\mathcal{M}(\ket{\psi}) = O(N)$.

\section{Implication to Bell's inequalities} \label{sec:bell}

Let us consider non-locality tests using the Bell-Clauser-Horne-Shimony-Holt (Bell-CHSH) inequality~\cite{clauser1969proposed} in our many-body spin setting with imprecise measurements.
The Bell-CHSH function is defined as
\begin{align}
\mathcal{B} = | E(a, b) - E(a, b') | + | E(a',b') + E(a',b) | .
\end{align}
Here, $E(a,b)$ is the correlation function of observables with dichotomy outcomes and $(a,a')$ and $(b,b')$ represent two different measurement set-ups for subsystems A and B, respectively.
The Bell theorem states that $\mathcal{B} \leq 2$ for local hidden variable theories.

To construct a dichotomy observable in our spin measurement set-up, we define the measurement operator for subsystem A as
\begin{align}
A(a) = \int_{-\infty}^{\infty} \dd{x} f(x) M^\sigma (a; x),
\end{align}
where $f(x)$ is an arbitrary function that gives either $1$ or $-1$ according to $x$. 
Likewise, we also define $B(b)$ for subsystem B as 
\begin{align}
B(b) = \int_{-\infty}^{\infty} \dd{x} g(y) N^\sigma (b; y).
\end{align}
As in the previous section, $\sigma$ denotes the degree of imprecision. 
Here, $a$ and $b$ parametrize the directions of collective spin measurements. 
In this set-up, a measurement setting $(a,b)$ can be transformed to others $(a',b), (a,b'), (a', b')$ using local unitary transforms.
The correlation function $E(a,b)$ for the Bell-CHSH function is then defined as $E(a,b) = \Tr [\rho_{AB} A(a) \otimes B(b)]$.
Under this setting, the following theorem holds.

\begin{theorem}
The Bell-CHSH function $\mathcal{B}$ for pure state $\ket{\psi}$ is bounded as
\begin{align*}
\mathcal{B} \leq 2 + 8\left\{1-\exp\left(-\frac{\mathcal{M}(\ket{\psi})}{16 \sigma^2}\right)\right\}.
\end{align*}
\end{theorem}
\begin{proof}
For product state $\rho_A \otimes \rho_B$, let $\widetilde{E}(a,b) = \Tr [(\rho_{A}\otimes \rho_B) A(a) \otimes B(b)]$ and $\widetilde{\mathcal{B}} =  | \widetilde{E}(a, b) - \widetilde{E}(a, b') | + | \widetilde{E}(a',b') + \widetilde{E}(a',b) | $.
Then,
\begin{align*}
&|E(a,b)-\widetilde{E}(a,b)| \\
&= \Bigl| \int \dd{x} \int  \dd{y} f(x) g(y) M^\sigma(x)\otimes N^\sigma(y) [\rho_{AB} - \rho_A \otimes \rho_B]\Bigr|\\
&\leq \int \dd{x} \int  \dd{y} \bigl| M^\sigma(x)\otimes N^\sigma(y) [\rho_{AB} - \rho_A \otimes \rho_B] \bigr| \\
&= 2\Delta_C 
\end{align*}
for arbitrary $a$ and $b$.
Using 
\begin{align*}
&\bigl| |E(a,b)-E(a,b')|-|\widetilde{E}(a,b)-\widetilde{E}(a,b')| \bigr| \\
&\leq |E(a,b)-\widetilde{E}(a,b)| + |E(a,b')-\widetilde{E}(a,b')|
\end{align*}
and
\begin{align*}
&\bigl| |E(a',b')+E(a',b)|-|\widetilde{E}(a',b')-\widetilde{E}(a',b)| \bigr| \\
&\leq |E(a',b')-\widetilde{E}(a',b')| + |E(a',b)-\widetilde{E}(a',b)|,
\end{align*}
the difference between the two Bell-CHSH functions is bounded as 
\begin{align*}
|\mathcal{B} - \widetilde{\mathcal{B}}| \leq 8 \Delta_C \leq 8\left\{1-\exp\left(-\frac{\mathcal{M}(\ket{\psi})}{16 \sigma^2}\right)\right\},
\end{align*}
where we have used Eq.~\eqref{eq:macro_bound}. This completes the proof as the Bell-CHSH function for a product state is bounded by $2$, i.e., $\widetilde{\mathcal{B}} \leq 2$.
\end{proof}

This theorem indicates that in order to observe a large violation of the Bell-CHSH inequality, $\mathcal{M}(\ket{\psi})$ should be sufficiently large and/or $\sigma$ should be sufficiently small.
This elucidates why previous studies have used macroscopic quantum superpositions to show a violation of the Bell-CHSH inequality or witness entanglement with imprecise measurements~\cite{jeong2009,lim2012,wang2013,jeong2014}.

\section{Conclusion} \label{sec:conc}

We have investigated bipartite entanglement in many-body spin systems in terms of the correlation in local measurements.
It turned out that the CLM is upper bounded by a function of quantum mutual information for general mixed states and there exist local POVMs that give a CLM larger than the linear entropy of a subsystem for pure states. 
As a realistic example, we have considered the case wherein local measurements are performed in the basis of a collective spin operator.
Under this restriction, while the CLM with appropriate spin directions properly captures the entanglement of spin squeezed states and the ground state of the Heisenberg XXZ model, it does not capture the correlation of Haar random states.
We have also considered the case of imprecise measurement and generalized the definition of the CLM accordingly.
It turned out that the measure of quantum macroscopicity gives a bound to the CLM with imprecise measurement and similarly to the Bell-CHSH function.
This analysis indicates that in order to observe a large violation of the Bell-CHSH inequality with many-body spin systems, one needs to prepare an entangled state with a large quantum macroscopicity.

\section*{Acknowledgement}
CYP thanks Hyukjoon Kwon for helpful discussions. This research was supported (in part) by the R\&D Convergence Program of NST (National Research Council of Science and Technology) of Republic of Korea (Grant No. CAP-15-08-KRISS).

\appendix



\section{Proof of Theorem~\ref{thm:fidelity-variance}}

Note that
\begin{widetext}
\begin{align*}
	\mathcal{F}(\ket{\psi}, \rho_{AB}')^2 &= \int_{D_X}dx \int_{D_Y} dy \, \bra{\psi}E^\sigma (x)\otimes F^\sigma(y)\ket{\psi}\bra{\psi}E^\sigma (x)\otimes F^\sigma(y)\ket{\psi}\\
	&=\sum_{\substack{i, i', j, j', \\ \mu_i, \mu_i' \mu_j, \mu_j'}} \exp\Bigl[ -\frac{(i-i')^2 + (j-j')^2}{8\sigma^2}\Bigr] |\braket{i, \mu_i,j,\mu_j| \psi}|^2|\braket{i', \mu_i',j',\mu_j'| \psi}|^2 \\
	& \geq  \exp\Bigl[ \sum_{\substack{i, i', j, j', \\ \mu_i, \mu_i' \mu_j, \mu_j'}} -\frac{(i-i')^2 + (j-j')^2}{8\sigma^2} |\braket{i, \mu_i,j,\mu_j| \psi}|^2|\braket{i', \mu_i',j',\mu_j'| \psi}|^2 \Bigr],
\end{align*}
\end{widetext}
where we have used $\int_{-\infty}^{\infty} \dd{x} \, p(x,i)^{1/2} p(x,i')^{1/2} = \exp[-(i-i')^2/(8\sigma^2)]$ in the second equality and the Jensen's inequality to obtain the last expression. Then the proof is completed as
\begin{widetext}
\begin{align*}
\sum_{\substack{i, i', j, j', \\ \mu_i, \mu_i' \mu_j, \mu_j'}} (i-i')^2 |\braket{i, \mu_i,j,\mu_j| \psi}|^2|\braket{i', \mu_i',j',\mu_j'| \psi}|^2 &= 2 \mathcal{V}_{\ket{\psi} }( S_A \otimes \Id),\\
\sum_{\substack{i, i', j, j', \\ \mu_i, \mu_i' \mu_j, \mu_j'}} (j-j')^2 |\braket{i, \mu_i,j,\mu_j| \psi}|^2|\braket{i', \mu_i',j',\mu_j'| \psi}|^2 &= 2 \mathcal{V}_{\ket{\psi} }( \Id \otimes S_B).
\end{align*}
\end{widetext}

\end{document}